\documentclass[aps,pra,twocolumn,10pt,superscriptaddress,nofootinbib,longbibliography]{revtex4-2}
\usepackage{amsmath}
\usepackage{amsthm}
\usepackage{txfonts}
\usepackage{graphicx}
\usepackage{bm}
\usepackage{tikz}
\usetikzlibrary{arrows.meta,positioning,decorations.pathmorphing}
\usepackage[colorlinks=true,linkcolor=blue,citecolor=blue,urlcolor=blue]{hyperref}
\usepackage{ragged2e}
 
\newtheorem{theorem}{Theorem}
\newtheorem{lemma}{Lemma}
\newtheorem{corollary}{Corollary}
\theoremstyle{remark}
\newtheorem{remark}{Remark}
 
\newcounter{smsec}

\newcommand{\smref}[1]{\hyperref[#1]{Appendix~\ref*{#1}}}
 
\newcommand{\tr}{\mathrm{Tr}}
\newcommand{\N}{\mathcal{N}}
\newcommand{\Pset}{\mathcal{P}}
\newcommand{\G}{\mathcal{G}}
\newcommand{\id}{\mathbb{I}}
\newcommand{\R}{\mathbb{R}}
\newcommand{\norm}[1]{\left\lVert #1 \right\rVert_1}
\newcommand{\ket}[1]{|#1\rangle}
\newcommand{\bra}[1]{\langle #1|}
 
\begin{document}
 
\title{Exact Catalysis Cannot Overcome the Gaussian-Steering Barrier for Remote Wigner Negativity}
 
\author{Sudip Chakrabarty}
\email{sudip27042000@gmail.com}
\affiliation{S. N. Bose National Centre for Basic Sciences, Block JD, Sector III, Salt Lake, Kolkata 700106, India}
 
\begin{abstract}
Gaussian steering sets a sharp threshold for creating Wigner negativity at a distance. A measurement by Bob can make Alice's part of a shared Gaussian state Wigner negative only if Alice can steer him with Gaussian measurements. We ask whether copies, ancillas, or catalysts can lift this requirement when Alice is restricted to Gaussian operations. We introduce the steered negativity, a multiplicative monotone whose logarithm equals the Gaussian steerability for Gaussian states. It bounds the negativity that Bob can herald and rules out activation by copies, by ancillas of Bob, and by catalysts without remote negativity of their own. For a strictly unsteerable Gaussian state and a single round of feed-forward to a Gaussian unitary of Alice, we further prove that even a Wigner-negative or steerable catalyst cannot help if it is returned exactly. Gaussian steering thus remains a barrier to remote Wigner negativity beyond the single-copy setting.
\end{abstract}
 
\maketitle
 
\section{Introduction} 
Wigner negativity is a hallmark of nonclassicality in quantum systems. It is necessary for continuous-variable quantum computation beyond efficient classical simulation~\cite{Mari2012} and underlies a resource theory in which Gaussian states and operations form the free
sector~\cite{Kenfack2004,Albarelli2018,Takagi2018}. In practice, Wigner
negativity is typically generated by a non-Gaussian operation such as photon
subtraction~\cite{Ourjoumtsev2006,Ra2020,Walschaers2021}. The non-Gaussian
operation, however, need not act where the negativity is ultimately needed.
When two parties share a Gaussian state, a measurement by Bob can make
Alice's subsystem Wigner negative precisely when Alice can steer Bob with
Gaussian measurements~\cite{Walschaers2020PRL,Walschaers2020PRXQ}. This
connection has been extended beyond Gaussian states~\cite{Walschaers2023},
studied for photon subtraction~\cite{Xiang2022}, and demonstrated
experimentally~\cite{Liu2022}. Related measurement-induced activation
mechanisms have also been studied for other forms of nonclassicality
~\cite{chakrabarty2026measurementinducedremoteactivationnonclassicality}.
Thus, Gaussian steering establishes a sharp boundary for remote Wigner
negativity in the single-copy setting.
 
Whether this boundary survives additional resources is less clear. For Bell
nonlocality, several copies can activate states that are useless one at a
time~\cite{Palazuelos2012,Navascues2011,Cavalcanti2011,Brunner2014},
and steering can likewise be superactivated~\cite{Quintino2016,Hsieh2016}.
Catalysis provides another route. A resource that is inaccessible on its own
can become useful with the help of a catalyst that is returned to its initial
state, as shown recently for Bell nonlocality~\cite{Bavaresco2025}. Such effects
occur in entanglement and other resource
theories~\cite{Jonathan1999,LipkaBartosik2021,Kondra2021,Datta2023,LipkaBartosik2024},
as well as for magic states~\cite{Campbell2011}. Catalytic generation of Wigner
negativity itself has been demonstrated in a cavity-QED setting, where a field
becomes Wigner negative while the atom returns exactly to its initial
state~\cite{deOliveiraJunior2024}. That mechanism relies on a non-Gaussian
light-matter interaction and is not remote. These observations raise a natural
question. \emph{Can multiple copies, ancillary systems, or a resourceful catalyst
overcome the Gaussian-steering barrier for remote Wigner negativity?}
 
We answer this question for one-sided Gaussian protocols, in which Alice is
restricted to Gaussian operations and Bob is unrestricted. Our starting point
is an operator-valued Wigner function that associates an operator on Bob's
system with each phase-space point of Alice. Its integrated trace norm defines
the \emph{steered negativity}. We show that it is multiplicative under tensor
products and nonincreasing under one-sided Gaussian protocols. For Gaussian
states, its logarithm coincides with the Gaussian
steerability~\cite{Kogias2015}. The steerability is known to be a monotone
when both parties are Gaussian~\cite{Lami2016}. Here it bounds the Wigner
negativity that an unrestricted Bob can herald. Multiplicativity
rules out activation by multiple copies, by arbitrary ancillas on Bob's side,
and by catalysts that are themselves remotely positive. These statements rest
on an elementary closure property of remotely positive states.
 
The stronger question concerns a catalyst that already contains the
resource. Such a catalyst can transfer its negativity if it is consumed. Can it
do so while being returned exactly? We prove that the answer is no for a
strictly unsteerable Gaussian resource and for single-round catalytic protocols in which
Bob's outcome is fed forward to a Gaussian unitary of Alice. This holds even
when the catalyst is Wigner negative, entangled, or steerable, and even when
exact return is required only on average. The proof rests on a rigidity of the
trace norm. Exact return forces every inequality behind the monotonicity of the
steered negativity to be saturated, and saturation leaves no negativity for
Alice. This is the main result of this work. Within this class of protocols, Gaussian steering is not merely a single-copy threshold.

\section{Setting}
\emph{Operator-valued Wigner function.} 
Alice holds $n$ bosonic modes with quadratures collected in a phase-space point $x$ of $\R^{2n}$. Bob holds an arbitrary quantum system. We set $\hbar$ to one, so that the vacuum has covariance matrix one half. The Wigner function of a state $\rho$ of $n$ modes is the expectation value of the Wigner operator \cite{Cahill1969},
\begin{align}
W_\rho(x) &= \tr\left[\rho\,\Delta(x)\right], &
\Delta(x) &= \pi^{-n} D(x)\,\Pi\, D(x)^\dagger ,
\end{align}
where $D(x)$ is the displacement operator and $\Pi$ the parity. For a bipartite state we apply the Wigner operator on Alice's side only and keep Bob's system as an operator,
\begin{align}
\sigma_B(x) &= \tr_A\left[\left(\Delta(x)\otimes \id_B\right)\rho_{AB}\right]. \label{eq:sigma}
\end{align}
This is a Hermitian trace-class operator on Bob's space for every $x$. It has three simple properties,
\begin{align}
\tr\,\sigma_B(x) &= W_{A}(x), \nonumber\\
\int dx\, \sigma_B(x) &= \rho_B, \nonumber\\
\tr\left[M \sigma_B(x)\right] &= p_M\, W_{A|M}(x). \label{eq:props}
\end{align}
In the last line $M$ is an effect of a measurement by Bob, $p_M$ is its probability, and $W_{A|M}$ is the Wigner function of Alice's conditional state. The last line shows that $\sigma_B$ contains everything that Bob's measurements can do to Alice's Wigner function. Throughout, we call a Wigner function or an operator positive when it has no negative part.
 
\begin{lemma}[Remote positivity criterion~\cite{Walschaers2023}]\label{lem:P}
No measurement by Bob can make Alice's conditional state Wigner negative if and only if
\begin{align}
\sigma_B(x) &\ge 0 \quad \text{for all } x. \label{eq:P}
\end{align}
\end{lemma}
 
We call a state that obeys Eq.~\eqref{eq:P} remotely positive and write $\Pset$ for the set of such states. For such a state $\sigma_B(x)$ is an ensemble of Bob's states labelled by a classical variable $x$ with the true probability density $W_A(x)$. This is a local hidden state model in which the hidden variable is Alice's phase-space point. Reference~\cite{Walschaers2023} phrases Eq.~\eqref{eq:P} as the physicality of the conditional Wigner function. We give a short proof in \smref{app:conv}.
 
For a Gaussian state with covariance matrix
\begin{align}
V_{AB} &= \begin{pmatrix} V_A & C \\ C^{T} & V_B \end{pmatrix},
\end{align}
Bob's operator $\sigma_B(x)$ is a displaced Gaussian operator with covariance matrix given by the Schur complement
\begin{align}
V_{B|A} &= V_B - C^{T} V_A^{-1} C . \label{eq:schur}
\end{align}
It is positive exactly when this matrix is a legitimate covariance matrix, which is the condition that Alice cannot steer Bob with Gaussian measurements \cite{Wiseman2007},
\begin{align}
V_{AB} + \frac{i}{2}\left(0_A \oplus \Omega_B\right) &\ge 0 . \label{eq:steer}
\end{align}
Here $\Omega_B$ is the symplectic form of Bob's modes. For Gaussian states, Lemma~\ref{lem:P} coincides with the Gaussian steering criterion~\cite{Walschaers2020PRL}. Throughout, steering refers to Gaussian measurements by Alice. Non-Gaussian measurements can reveal steering in states that obey Eq.~\eqref{eq:steer}~\cite{Ji2016}, but Alice does not have them here.

\emph{Allowed operations.} 
Wigner negativity is trivial to create if Alice has a non-Gaussian operation, so Alice is restricted. She can apply Gaussian unitaries, add modes in Gaussian states, discard modes, and perform homodyne or heterodyne measurements, all conditioned on classical information. Bob is not restricted. He can add any ancilla and apply any instrument, that is, a set of completely positive maps $\mathcal{E}_k$ whose sum preserves the trace. Both can send classical messages. We call a protocol built from these steps a one-sided Gaussian protocol. It ends with a classical record $k$ that occurs with probability $p_k$ and a state $\rho_k$ of the modes that Alice keeps. Remote generation of Wigner negativity means that some $\rho_k$ has a Wigner function with a negative part.
 
All results of this work are relative to this class. It contains the protocols of Refs.~\cite{Walschaers2020PRL,Walschaers2020PRXQ,Xiang2022,Liu2022}, in which Alice does nothing, and adds Gaussian processing, feed-forward, and communication in both directions. We take the record $k$ to be discrete. For a continuous record the sums become integrals, and statements for every outcome hold for almost every outcome.

\section{Steered negativity} 
For a single system the size of the negative part of the Wigner function is measured by its integrated absolute value \cite{Kenfack2004,Albarelli2018}. It can also be certified operationally from few-copy measurements of moments of the Wigner function~\cite{chakrabarty2025,chakrabarty2026operationalcertificationnonclassicalityarbitrary}. Equation~\eqref{eq:props} suggests the analogous quantity for the remote task, with the absolute value replaced by the trace norm on Bob's space,
\begin{align}
\N(\rho_{AB}) &= \int dx \, \norm{\sigma_B(x)} . \label{eq:N}
\end{align}
We call it the steered negativity of Alice by Bob. It relates to the Wigner function in the way the negativity of Vidal and Werner relates to the partial transpose \cite{Vidal2002}. The next theorem shows that it is a monotone for one-sided Gaussian protocols. The proof is in \smref{app:thm1}.
 
\begin{theorem}[Steered negativity]\label{thm:mono}
The steered negativity has the following properties.
\begin{enumerate}
\item[(i)] It is at least one, and it equals one if and only if the state is remotely positive.
\item[(ii)] It is multiplicative,
\begin{align}
\N(\rho_{AB}\otimes\tau_{A'B'}) &= \N(\rho_{AB})\,\N(\tau_{A'B'}). \label{eq:mult}
\end{align}
\item[(iii)] It is convex.
\item[(iv)] It does not increase on average under one-sided Gaussian protocols,
\begin{align}
\sum_k p_k\, \N(\rho_k) &\le \N(\rho_{AB}). \label{eq:mono}
\end{align}
\item[(v)] If Bob's system is trivial it reduces to the integrated absolute value of the Wigner function,
\begin{align}
\N(\rho_A) &= \int dx\, |W_A(x)| .
\end{align}
\end{enumerate}
\end{theorem}
 
Property (v) together with Eq.~\eqref{eq:mono} gives a bound on everything that a protocol can deliver. We also obtain a bound for each outcome separately. For this we split $\sigma_B(x)$ into its positive and negative parts and define a positive operator on Bob's side,
\begin{align}
\sigma_B(x) &= \sigma_B^{+}(x) - \sigma_B^{-}(x), &
\Sigma_B &= \int dx\, \sigma_B^{-}(x). \label{eq:Sigma}
\end{align}
Its trace is fixed by the steered negativity,
\begin{align}
\tr\,\Sigma_B &= \frac{\N(\rho_{AB})-1}{2}. \label{eq:trSigma}
\end{align}
 
\begin{corollary}[Negativity budget]\label{cor:budget}
Let a one-sided Gaussian protocol act on $\rho_{AB}$ and produce the Wigner function $W_k$ on Alice's side with probability $p_k$. Then
\begin{align}
\sum_k p_k \int dx\,|W_k(x)| &\le \N(\rho_{AB}). \label{eq:budget}
\end{align}
If the protocol is a single measurement by Bob with effects $M_k$, then for every outcome
\begin{align}
p_k\left(\int dx\,|W_k(x)| - 1\right) &\le 2\,\tr\left[M_k \Sigma_B\right] \le \N(\rho_{AB}) - 1. \label{eq:herald}
\end{align}
\end{corollary}
 
Equation~\eqref{eq:herald} is a trade-off between how negative a heralded state is and how often it is heralded. The operator $\Sigma_B$ tells Bob which of his outcomes can herald negativity at all. The bound in Eq.~\eqref{eq:budget} is attained by states for which the operators $\sigma_B(x)$ commute for different $x$, for example states that are classical on Bob's side. Bob then measures in the common eigenbasis. For Gaussian states the steered negativity can be evaluated in closed form. The proof is in \smref{app:thm2}.
 
\begin{theorem}[Gaussian states]\label{thm:gauss}
Let $\rho_{AB}$ be a Gaussian state and let $\nu_j$ be the symplectic eigenvalues of the matrix $V_{B|A}$ of Eq.~\eqref{eq:schur}. Then
\begin{align}
\N(\rho_{AB}) &= \prod_j \max\left\{1, \frac{1}{2\nu_j}\right\}, \label{eq:Ngauss}\\
\ln \N(\rho_{AB}) &= \G_{A\to B}(\rho_{AB}), \label{eq:steerability}
\end{align}
where $\G_{A\to B}$ is the Gaussian steerability of Ref.~\cite{Kogias2015}.
\end{theorem}
 
The Gaussian steerability was introduced as a measure of how far the matrix $V_{B|A}$ is from being a covariance matrix. It does not increase under Gaussian operations of both parties~\cite{Lami2016}. Theorem~\ref{thm:gauss} and Corollary~\ref{cor:budget} show that it also constrains a non-Gaussian Bob. Its exponential is an upper bound on the average negativity that Bob can herald on Alice's side, whatever he measures and whatever Gaussian processing Alice adds. This is a bound, and we do not claim that it is attained for Gaussian states. By Jensen's inequality the average logarithmic negativity of Alice's heralded states is bounded by the steerability itself,
\begin{align}
\sum_k p_k \ln \int dx\,|W_k(x)| &\le \G_{A\to B}(\rho_{AB}). \label{eq:logbound}
\end{align}
Reference~\cite{Xiang2022} found that for photon subtraction a more steerable state can give less negativity, and concluded that purities and not the steerability govern that protocol. Equation~\eqref{eq:logbound} is consistent with this finding and complements it. The steerability does not predict the negativity of a given heralded state. It bounds all of them together.
 
As an illustration we take a two-mode squeezed vacuum with squeezing parameter $r$ whose mode on Alice's side has passed a channel with transmission $\eta$. Equation~\eqref{eq:Ngauss} gives
\begin{align}
\N &= \max\left\{1, \frac{\eta\cosh 2r + 1 - \eta}{\eta + (1-\eta)\cosh 2r}\right\}. \label{eq:lossy}
\end{align}
The steered negativity exceeds one when $\eta$ is larger than one half, which is the known steering threshold \cite{Walschaers2020PRL}. Figure~\ref{fig:budget} compares this budget with the negativity that Bob heralds by counting photons. A numerical integration of Eq.~\eqref{eq:N} in the Fock basis reproduces Eq.~\eqref{eq:lossy}. For the parameters of the figure and no loss the budget for the excess negativity is $0.543$, while photon counting reaches $0.108$ on average. The bound is therefore not attained by photon counting. The reason can be seen from Eq.~\eqref{eq:herald}. To saturate it Bob would have to measure in the eigenbasis of $\sigma_B(x)$, and for a Gaussian state this basis is displaced by an amount that depends on $x$. Whether another measurement comes closer is open.

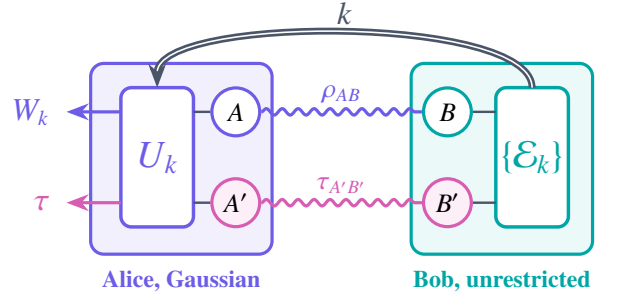
\begin{figure}[t]
\centering

\begin{tikzpicture}[
    >=Stealth,
    scale=0.8,
    every node/.style={font=\small},
    line width=0.9pt
]

\definecolor{alicepurple}{RGB}{108,92,231}
\definecolor{bobteal}{RGB}{0,166,166}
\definecolor{catalystpink}{RGB}{214,93,177}
\definecolor{softpurple}{RGB}{242,240,253}
\definecolor{softteal}{RGB}{235,249,249}
\definecolor{softpink}{RGB}{252,239,248}
\definecolor{darkslate}{RGB}{74,85,104}

\draw[
    rounded corners=5pt,
    fill=softpurple,
    draw=alicepurple,
    line width=1.2pt
]
(-0.3,-0.45) rectangle (2.7,2.7);

\draw[
    rounded corners=5pt,
    fill=softteal,
    draw=bobteal,
    line width=1.2pt
]
(5.0,-0.45) rectangle (8.0,2.7);

\node[
    font=\bfseries,
    text=alicepurple
]
at (1.2,-0.9)
{Alice, Gaussian};

\node[
    font=\bfseries,
    text=bobteal
]
at (6.5,-0.9)
{Bob, unrestricted};

\node[
    circle,
    draw=alicepurple,
    fill=white,
    line width=1.2pt,
    inner sep=1.5pt,
    minimum size=18pt,
    font=\normalsize\bfseries
]
(A) at (2.1,1.9) {$A$};

\node[
    circle,
    draw=bobteal,
    fill=white,
    line width=1.2pt,
    inner sep=1.5pt,
    minimum size=18pt,
    font=\normalsize\bfseries
]
(B) at (5.6,1.9) {$B$};

\node[
    circle,
    draw=catalystpink,
    fill=softpink,
    line width=1.2pt,
    inner sep=1.5pt,
    minimum size=18pt,
    font=\normalsize\bfseries
]
(Ap) at (2.1,0.4) {$A'$};

\node[
    circle,
    draw=catalystpink,
    fill=softpink,
    line width=1.2pt,
    inner sep=1.5pt,
    minimum size=18pt,
    font=\normalsize\bfseries
]
(Bp) at (5.6,0.4) {$B'$};

\draw[
    decorate,
    decoration={
        snake,
        amplitude=1.2pt,
        segment length=6pt
    },
    draw=alicepurple,
    line width=1.2pt
]
(A)
to node[
    above,
    font=\normalsize\bfseries,
    text=alicepurple
] {$\rho_{AB}$}
(B);

\draw[
    decorate,
    decoration={
        snake,
        amplitude=1.2pt,
        segment length=6pt
    },
    draw=catalystpink,
    line width=1.2pt
]
(Ap)
to node[
    above,
    font=\normalsize\bfseries,
    text=catalystpink
] {$\tau_{A'B'}$}
(Bp);

\draw[
    rounded corners=4pt,
    fill=white,
    draw=bobteal,
    line width=1.3pt
]
(6.4,0.0) rectangle (7.6,2.3);

\node[
    font=\Large\bfseries,
    text=bobteal
]
at (7.0,1.15)
{$\{\mathcal{E}_k\}$};

\draw[
    rounded corners=4pt,
    fill=white,
    draw=alicepurple,
    line width=1.3pt
]
(0.2,0.0) rectangle (1.4,2.3);

\node[
    font=\Large\bfseries,
    text=alicepurple
]
at (0.8,1.15)
{$U_k$};

\draw[
    draw=darkslate,
    line width=0.9pt
]
(B) -- (6.4,1.9);

\draw[
    draw=darkslate,
    line width=0.9pt
]
(Bp) -- (6.4,0.4);

\draw[
    draw=darkslate,
    line width=0.9pt
]
(A) -- (1.4,1.9);

\draw[
    draw=darkslate,
    line width=0.9pt
]
(Ap) -- (1.4,0.4);

\draw[
    ->,
    double,
    draw=darkslate,
    line width=1.0pt
]
(7.0,2.3)
..
controls (7.0,3.5) and (0.8,3.5) ..
node[
    above,
    font=\large\bfseries,
    text=darkslate
] {$k$}
(0.8,2.3);

\draw[
    ->,
    draw=alicepurple,
    line width=1.2pt
]
(0.2,1.9) -- (-0.7,1.9);

\node[
    left,
    xshift=-3pt,
    font=\large\bfseries,
    text=alicepurple
]
at (-0.7,1.9) {$W_k$};

\draw[
    ->,
    draw=catalystpink,
    line width=1.2pt
]
(0.2,0.4) -- (-0.7,0.4);

\node[
    left,
    xshift=-3pt,
    font=\large\bfseries,
    text=catalystpink
]
at (-0.7,0.4) {$\tau$};

\end{tikzpicture}

\caption{\justifying
\small
\emph{The catalytic scenario.}
Alice and Bob share a strictly unsteerable Gaussian state $\rho_{AB}$ and a catalyst
$\tau_{A'B'}$. Bob applies an arbitrary instrument and sends the outcome $k$.
Alice applies a Gaussian unitary that may depend on $k$.
The catalyst must be returned to its initial state.
Theorem~\ref{thm:nocat} states that Alice's output $W_k$ is then Wigner positive.
}
\label{fig:scheme}
\end{figure}

\begin{figure*}[t]
\centering
\includegraphics[width=0.86\textwidth]{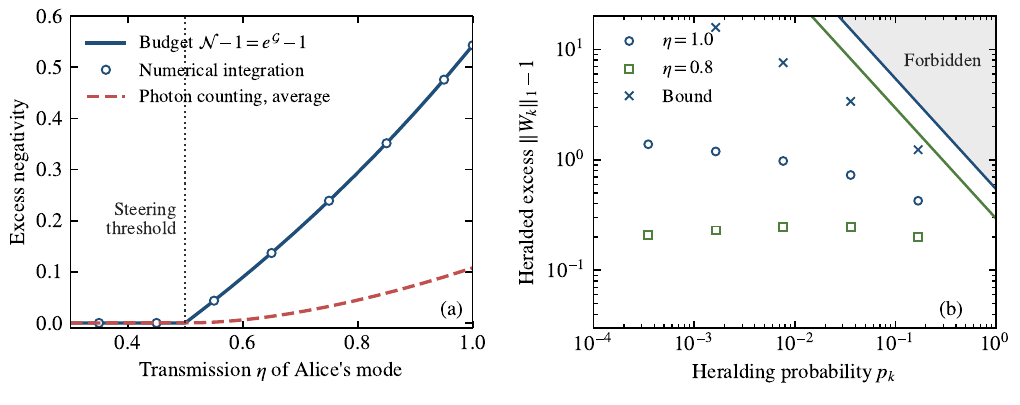}
\caption{Negativity budget for a two-mode squeezed vacuum with $r$ equal to $0.5$ and transmission $\eta$ on Alice's mode. (a) Excess steered negativity from Eq.~\eqref{eq:lossy} (solid line), from numerical integration of Eq.~\eqref{eq:N} (circles), and the average excess negativity that Bob heralds by photon counting (dashed line). (b) Heralding probability and excess negativity of Alice's state for Bob's photon-number outcomes one to six (from right to left). The solid lines are the bound on the right of Eq.~\eqref{eq:herald}, and the shaded region is excluded for all measurements when there is no loss. Crosses show the outcome-resolved bound in the middle of Eq.~\eqref{eq:herald} for the lossless case.}
\label{fig:budget}
\end{figure*}

\section{No activation} 
Bell nonlocality and steering can be activated because the set of states with a local model is not closed under tensor products \cite{Palazuelos2012,Quintino2016}. The set $\Pset$ is closed, since a tensor product of positive operators is positive, and it stays closed under everything that Alice and Bob are allowed to do. The argument is elementary and close in spirit to the simulation of Wigner-positive circuits~\cite{Mari2012}, but its consequence for the remote task has not been stated. The comparison with nonlocality has a limit. Bell nonlocality refers to unrestricted parties, while remote Wigner negativity is defined relative to a restriction on Alice. The following theorem is a direct consequence of Theorem~\ref{thm:mono} and holds for the full class of one-sided Gaussian protocols.
 
\begin{theorem}[No activation]\label{thm:noact}
Let $\rho_{AB}$ and $\tau_{A'B'}$ be remotely positive. Then every state that a one-sided Gaussian protocol produces from $\rho_{AB}\otimes\tau_{A'B'}$ on Alice's side has a positive Wigner function.
\end{theorem}
 
\begin{proof}
By Theorem~\ref{thm:mono}(i) both states have unit steered negativity. By Eq.~\eqref{eq:mult} so does their product. The integrated absolute value of each $W_k$ is at least one, and by Eq.~\eqref{eq:budget} its average is at most one. It therefore equals one for every $k$, which means that $W_k$ is positive.
\end{proof}
 
Three special cases deserve to be stated. First, no number of copies of an unsteerable Gaussian state allows remote Wigner negativity, even if Bob measures all his modes jointly and Alice interferes all of hers. There is no superactivation.
Second, nothing that Bob holds can help. A state of the form $\rho_{AB}\otimes\omega_{B'}$ is remotely positive for every $\omega_{B'}$.
Third, a catalyst can help only if it carries the resource itself. The excluded catalysts include all shared Gaussian states that are unsteerable from Alice to Bob, all separable states whose components on Alice's side are Wigner positive, and every Wigner-positive state held by Alice alone. For such catalysts the result holds whether or not the catalyst is returned.
 
The third point is in contrast with Bell nonlocality, where a Bell-local catalyst activates a Bell-local state \cite{Bavaresco2025}. It leaves one possibility open. A catalyst outside $\Pset$ carries remote negativity or local negativity. Alice can of course take it. If the catalyst is a single photon in a mode of Alice, she can swap it into her output mode. The catalyst is then consumed. The question that remains is whether such a catalyst can transfer negativity to Alice and still be returned intact. This is the catalytic question proper, and we answer it next.

\section{No catalysis} 
We now consider a catalytic subclass of one-sided Gaussian protocols, sketched in Fig.~\ref{fig:scheme}. Alice holds $n$ modes $A$ and one catalyst mode $A'$. Bob holds modes $B$ and an arbitrary system $B'$. The initial state is $\rho_{AB}\otimes\tau_{A'B'}$. Bob applies an instrument with operations $\mathcal{E}_k$ that map $BB'$ to $B'$ and sends $k$ to Alice. Alice applies a Gaussian unitary $U_k$ to $AA'$. With probability $p_k$ the final state of $AA'B'$ is $\omega_k$. The protocol is catalytic if the catalyst is returned exactly,
\begin{align}
\sum_k p_k\, \tr_A\, \omega_k &= \tau_{A'B'}. \label{eq:return}
\end{align}
This form covers more than it seems. Gaussian ancillas of Alice are included in $\rho_{AB}$ as extra modes of $A$. Modes of $A$ that Alice discards or measures by homodyne or heterodyne detection at the end do not matter, because a Wigner-positive state stays Wigner positive under both. \smref{app:thm4} shows that Alice may also start with homodyne or heterodyne measurements on some of her modes and send the results to Bob. The theorem does not cover protocols with further rounds, in which Alice measures modes that have already interacted with the catalyst and Bob acts again on the result.
 
We ask for the return only on average over $k$, which is the weakest form of the requirement. Correlations between the catalyst and Alice's output are allowed. The catalyst is arbitrary. It may be Wigner negative on Alice's side, entangled, or steerable. We only assume that its steered negativity is finite.
 
The theorem needs one assumption on the Gaussian state. We say that it is strictly unsteerable from Alice to Bob if Eq.~\eqref{eq:steer} holds with a strict inequality,
\begin{align}
V_{AB} + \frac{i}{2}\left(0_A \oplus \Omega_B\right) &> 0 . \label{eq:strict}
\end{align}
This holds for every unsteerable Gaussian state that is not on the boundary of the unsteerable set, and for every Gaussian state held by Alice alone. On the boundary the matrix has a zero eigenvalue. An example is the state of Eq.~\eqref{eq:lossy} at transmission one half. The boundary is not covered by the theorem, and it cannot be reached by a limit, because exact return of the catalyst is not a stable condition (see \smref{app:thm4}). Systems of Bob that are not correlated with Alice need not be included in $\rho_{AB}$, since Bob can prepare them himself.
 
\begin{theorem}[No catalysis]\label{thm:nocat}
Let $\rho_{AB}$ be a Gaussian state that is strictly unsteerable from Alice to Bob, and let $\tau_{A'B'}$ be any catalyst with a single mode $A'$ on Alice's side and finite steered negativity. Consider a catalytic protocol in which Bob applies an arbitrary instrument on $BB'$ and sends the outcome $k$ to Alice, Alice applies a Gaussian unitary $U_k$ on $AA'$, and the catalyst is returned exactly in the sense of Eq.~\eqref{eq:return}. Then, for every outcome $k$ that occurs with nonzero probability, the state of Alice's modes $A$ has a positive Wigner function.
\end{theorem}
 
The theorem says nothing about steerable states, for which remote negativity needs no catalyst. The proof is in \smref{app:thm4}. It has three steps.

 The first step is conservation. The steered negativity of the initial state equals that of the catalyst, because of Eq.~\eqref{eq:mult} and because $\rho_{AB}$ is remotely positive. It cannot grow during the protocol, and discarding $A$ at the end cannot make it grow either. The final catalyst is the initial one, so every inequality in this chain is an equality.
 
 The second step is rigidity. The last inequality of the chain is a triangle inequality for the trace norm, taken over all contributions that are added up when $A$ is traced out. Equality holds only if all contributions have their positive parts in one common subspace and their negative parts in the orthogonal complement. Once Alice's unitary couples $A$ to the catalyst mode, the Gaussian noise that $A$ injects has full support, and the contributions come from a whole line or from the whole plane of the catalyst's phase space.
 
 The third step uses a basic fact about Wigner functions. The integral of a Wigner function along a line is a probability density, and over the plane it is a probability. In operator form, the integral of the catalyst's operator-valued Wigner function over a line or the plane is a positive operator. A family of operators with aligned signs whose integral is positive has no negative part. So for every outcome in which Alice couples her modes to the catalyst, Bob's operation has already removed all negative parts, and Alice's output is Wigner positive. For the other outcomes the catalyst plays no role.
 
A simple example shows the mechanism. Let the catalyst be a single photon in mode $A'$ and let Alice mix it with a vacuum mode $A$ on a beam splitter with reflectivity $R$. Mode $A$ leaves in the state of a photon that has passed a channel with transmission $R$, and the catalyst is left as a photon with transmission $1-R$. For a single photon after transmission $\eta$ of at least one half, the integrated absolute value of the Wigner function is
\begin{align}
\int dx\,|W(x)| &= 4\eta\, e^{-(2\eta-1)/2\eta} - 1 . \label{eq:lossyphoton}
\end{align}
For smaller $\eta$ the Wigner function is positive and the integral is one. Alice's mode is negative only if $R$ exceeds one half, and then the catalyst is no longer negative. This is transfer. The negativity is moved from one mode to the other, the catalyst is degraded, and the negativity is never present in both. Theorem~\ref{thm:nocat} shows that no choice of catalyst, Gaussian unitary, or operation by Bob can change this within the catalytic protocols considered here. Three remarks put the theorem in context.
 
\begin{remark}\label{rem:broad}
Broadcasting means that the resource in a state is copied onto a second system while the first keeps it \cite{Barnum1996,Lostaglio2019,Marvian2019}. Our task meets it when Bob is absent. Theorem~\ref{thm:nocat} then states that a Gaussian unitary cannot make modes in a Gaussian state Wigner negative while a second system keeps its state, which is the Wigner-negativity part of the no-broadcasting theorem of Ref.~\cite{Chatterjee2025}, obtained here by a different argument. Our theorem adds the remote setting, with a correlated partner who performs arbitrary non-Gaussian operations and classical feed-forward.
\end{remark}
 
\begin{remark}\label{rem:emb}
Exactness matters. Let Alice hold $m$ catalyst modes in an equal mixture of the states in which the first $j$ modes contain one photon each and the rest are empty, with $j$ running from one to $m$. A cyclic shift of all modes puts a single photon into $A$ with certainty and changes the catalyst by $1/m$ in trace distance. This is the usual embezzling phenomenon \cite{vanDam2003,LipkaBartosik2024}. A catalyst that is only approximately returned can yield a fixed amount of negativity with an error that vanishes with its size. The theorem concerns exact return, as in Refs.~\cite{Bavaresco2025,Jonathan1999}, and gives no robustness bound for a catalyst that is slightly disturbed.
\end{remark}
 
\begin{remark}\label{rem:multi}
\smref{app:multi} extends the theorem to several catalyst modes on Alice's side when the coupling obeys a coisotropic condition, which full-rank couplings satisfy. Degenerate couplings, in which some catalyst modes act as spectators, are not covered.
\end{remark}

\section{Discussion}
We have shown that the Gaussian-steering barrier to remote Wigner negativity persists under the operations considered here. Copies and arbitrary ancillas on Bob's side cannot activate negativity from a remotely positive resource, and an exactly returned resourceful catalyst cannot generate remote Wigner negativity from a strictly unsteerable Gaussian state within the catalytic protocols considered here.
A consumed catalyst can still transfer its negativity. This does not rule out catalytic Wigner-negativity generation in general, which can occur with non-Gaussian interactions~\cite{deOliveiraJunior2024}.
The result has a simple operational origin. For a remotely positive state, Alice's phase-space variable provides a positive model for Bob's conditional state. Gaussian operations transform this variable only through linear maps and averaging, so positivity is preserved under tensor products and the allowed protocols. The steered negativity quantifies the failure of this model. Its multiplicativity gives the many-copy no-go result, and the rigidity of its conservation gives the catalytic one.
 
Several questions remain open. In particular, the optimal average negativity that Bob can herald from a steerable Gaussian state is unknown, as is whether a catalyst can amplify it when the steering constraint is already satisfied. A robust version of Theorem~\ref{thm:nocat}, relating the generated negativity to the disturbance of an approximately returned catalyst, would also be valuable for experiments. Extensions to boundary states, to multi-round protocols, to degenerate multimode couplings, and to discrete Wigner representations~\cite{Veitch2014} provide further directions.
 
\bibliography{remote_negativity}
 
\setcounter{equation}{0}
\setcounter{lemma}{0}
\renewcommand{\theequation}{S\arabic{equation}}
\renewcommand{\theHequation}{S\arabic{equation}}
\renewcommand{\thelemma}{S\arabic{lemma}}
\renewcommand{\theHlemma}{S\arabic{lemma}}
\setcounter{figure}{0}
\renewcommand{\thefigure}{S\arabic{figure}}
\renewcommand{\theHfigure}{S\arabic{figure}}

\appendix

\section{Conventions and proof of Lemma~\ref{lem:P}}\label{app:conv}
 
The Wigner operators of $n$ modes obey
\begin{align}
\int dx\,\Delta(x) &= \id, \qquad
U^\dagger \Delta(x)\, U = \Delta\!\left(S^{-1}(x-d)\right), \label{eq:cov}
\end{align}
where $U$ is a Gaussian unitary with symplectic matrix $S$ and displacement $d$. The second relation says that Gaussian unitaries act on Wigner functions as affine changes of variables with unit Jacobian. The properties in Eq.~\eqref{eq:props} follow from Eq.~\eqref{eq:sigma} and the first relation.
 
To prove Lemma~\ref{lem:P}, assume first that $\sigma_B(x)$ is positive for all $x$. Then the last line of Eq.~\eqref{eq:props} is nonnegative for every effect $M$, so no conditional Wigner function has a negative part. Conversely, assume that $\sigma_B(x_0)$ has a negative eigenvalue with eigenvector $\ket{\phi}$. Let $M$ be the projector on $\ket{\phi}$. Its probability is not zero, since otherwise $\rho_{AB}$ would have no support on $\ket{\phi}$ and $\sigma_B(x_0)\ket{\phi}$ would vanish. By Eq.~\eqref{eq:props} the conditional Wigner function is negative at $x_0$.
 
\section{Proof of Theorem~\ref{thm:mono} and Corollary~\ref{cor:budget}}\label{app:thm1}
 
(i) For a Hermitian operator the trace norm is at least the trace, with equality only for positive operators. With the first line of Eq.~\eqref{eq:props},
\begin{align}
\N(\rho_{AB}) &\ge \int dx\, \tr\,\sigma_B(x) = \int dx\, W_A(x) = 1 .
\end{align}
Equality holds if and only if $\sigma_B(x)$ is positive for almost all $x$. Displacements act continuously on trace-class operators, so $\sigma_B(x)$ is continuous in $x$ and is then positive for all $x$.
 
(ii) The Wigner operator of the modes $AA'$ is the product of the two Wigner operators. Hence the operator-valued Wigner function of a product state is the tensor product $\sigma_B(x)\otimes\sigma_{B'}(x')$, and the trace norm of a tensor product is the product of the trace norms.
 
(iii) The map from the state to $\sigma_B(x)$ is linear and the trace norm is convex.
 
(iv) We check each step of a protocol. We write the state before the step as a classical mixture over the record so far, and by (iii) and linearity it suffices to treat a single branch.
 
A Gaussian unitary on Alice's side changes $\sigma_B(x)$ into $\sigma_B(S^{-1}(x-d))$ by Eq.~\eqref{eq:cov}. The integral in Eq.~\eqref{eq:N} is unchanged.
 
Adding modes in a state with positive Wigner function $g$ changes $\sigma_B(x)$ into $g(x_2)\,\sigma_B(x_1)$. The integral is unchanged because $g$ is a probability density.
 
Discarding the modes with coordinates $x_2$ gives
\begin{align}
\int dx_1 \norm{\int dx_2\, \sigma_B(x_1,x_2)} &\le \int dx_1 dx_2 \norm{\sigma_B(x_1,x_2)} .
\end{align}
 
A homodyne or heterodyne measurement on the modes with coordinates $x_2$ has effects whose Wigner functions are nonnegative. We normalize them to weights $m_j(x_2)$ that add up to one. After outcome $j$ the unnormalized operator is the integral of $m_j(x_2)\sigma_B(x_1,x_2)$ over $x_2$. Summing the norms over $j$ and using the triangle inequality gives again the right-hand side of the previous equation.
 
An instrument on Bob's side with operations $\mathcal{E}_k$ changes $\sigma_B(x)$ into $\mathcal{E}_k(\sigma_B(x))$. For a Hermitian operator $X$ with positive and negative parts $X^{\pm}$,
\begin{align}
\sum_k \norm{\mathcal{E}_k(X)} &\le \sum_k \tr\,\mathcal{E}_k(X^{+}) + \tr\,\mathcal{E}_k(X^{-}) \le \norm{X}, \label{eq:instr}
\end{align}
because the sum of the operations does not increase the trace. Classical communication only makes later steps depend on earlier records, which is covered by treating each branch separately.
 
(v) For trivial $B$ the operator $\sigma_B(x)$ is the number $W_A(x)$.
 
The proof of (iv) uses only that Alice's ancillas and measurement effects have positive Wigner functions. The theorem therefore remains valid if Alice is given any Wigner-positive states and measurements.
 
For the corollary, Eq.~\eqref{eq:budget} is Eq.~\eqref{eq:mono} combined with (v). For Eq.~\eqref{eq:trSigma}, the trace norm of $\sigma_B(x)$ is its trace plus twice the trace of its negative part, and we integrate over $x$. For Eq.~\eqref{eq:herald}, the last line of Eq.~\eqref{eq:props} gives
\begin{align}
p_k W_k(x) &= \tr\left[M_k\sigma_B^{+}(x)\right] - \tr\left[M_k \sigma_B^{-}(x)\right] \nonumber\\
&\ge - \tr\left[M_k \sigma_B^{-}(x)\right].
\end{align}
The negative part of $p_kW_k$ is therefore at most the last trace. Its integral is half the left-hand side of Eq.~\eqref{eq:herald}, and the integral of the trace is the middle expression divided by two. The second inequality in Eq.~\eqref{eq:herald} follows from Eq.~\eqref{eq:trSigma} because $M_k$ is at most the identity.
 
\section{Proof of Theorem~\ref{thm:gauss}}\label{app:thm2}
 
We can take the first moments to vanish. The map from operators on Bob's modes to their Wigner functions is linear and one to one. By Eq.~\eqref{eq:sigma}, $\sigma_B(x)$ is the operator whose Wigner function is $W_{AB}(x,\cdot)$. A Gaussian function of two sets of variables factorizes into its marginal and its conditional,
\begin{align}
W_{AB}(x,x_B) &= W_A(x)\; G_{V_{B|A}}\!\left(x_B - C^{T}V_A^{-1}x\right),
\end{align}
where $G_V$ is the normalized centered Gaussian with covariance matrix $V$. Therefore
\begin{align}
\sigma_B(x) &= W_A(x)\; D(C^{T}V_A^{-1}x)\,\gamma\, D(C^{T}V_A^{-1}x)^\dagger ,
\end{align}
where $\gamma$ is the operator with Wigner function $G_{V_{B|A}}$. The trace norm is invariant under the displacement and $W_A$ integrates to one, so the steered negativity is the trace norm of $\gamma$.
 
The matrix $V_{B|A}$ is real, symmetric, and positive definite. By Williamson's theorem \cite{Williamson1936} there is a symplectic matrix that brings it to the diagonal form with entries $\nu_j$, each appearing twice. The corresponding Gaussian unitary turns $\gamma$ into a tensor product of single-mode operators $\gamma_{\nu_j}$ with isotropic Gaussian Wigner functions of variance $\nu_j$. For $\nu$ at least one half this operator is a thermal state. Its expansion in the Fock basis,
\begin{align}
\gamma_\nu &= \frac{2}{2\nu+1}\sum_{m=0}^{\infty}\left(\frac{2\nu-1}{2\nu+1}\right)^{m} \ket{m}\bra{m}, \label{eq:gamma}
\end{align}
remains valid for all positive $\nu$, since the generating function of the Wigner functions of the Fock states converges whenever the ratio in Eq.~\eqref{eq:gamma} has modulus below one. For $\nu$ below one half the ratio is negative, the odd Fock states carry negative weights, and
\begin{align}
\norm{\gamma_\nu} &= \frac{2}{2\nu+1}\,\frac{1}{1-\frac{1-2\nu}{1+2\nu}} = \frac{1}{2\nu}.
\end{align}
For $\nu$ at least one half the norm is one. The product over the modes gives Eq.~\eqref{eq:Ngauss}. The Gaussian steerability of Ref.~\cite{Kogias2015} is the sum of $-\ln(2\nu_j)$ over the symplectic eigenvalues below one half in our units, which is the logarithm of Eq.~\eqref{eq:Ngauss}.
 
For Eq.~\eqref{eq:lossy} the blocks of the covariance matrix are proportional to the identity up to a sign in $C$, with
\begin{align}
V_A &= \tfrac{1}{2}(\eta\cosh 2r + 1-\eta), \nonumber\\
V_B &= \tfrac{1}{2}\cosh 2r, \nonumber\\
C^{T}C &= \tfrac{1}{4}\,\eta\sinh^2 2r ,
\end{align}
and Eq.~\eqref{eq:schur} gives the single symplectic eigenvalue
\begin{align}
\nu &= \frac{\eta + (1-\eta)\cosh 2r}{2\left(\eta\cosh 2r + 1 - \eta\right)}.
\end{align}
 
\section{Proof of Theorem~\ref{thm:nocat}}\label{app:thm4}
 
The protocol is sketched in Fig.~\ref{fig:scheme}. We need three lemmas, numbered S1 to S3, to prove Theorem~\ref{thm:nocat}.
 
\begin{lemma}[Aligned signs]\label{lem:align}
Let $X(\omega)$ be Hermitian trace-class operators that depend measurably on a point $\omega$ of a measure space with measure $\mu$, and assume that the integral of their trace norms is finite. If
\begin{align}
\norm{\int d\mu(\omega)\, X(\omega)} &= \int d\mu(\omega) \norm{X(\omega)}, \label{eq:equal}
\end{align}
then there is a projector $P$ such that for almost all $\omega$ the positive part $X^{+}(\omega)$ and the negative part $X^{-}(\omega)$ obey
\begin{align}
P X^{+}(\omega) P &= X^{+}(\omega), &
(\id-P) X^{-}(\omega) (\id-P) &= X^{-}(\omega). \label{eq:aligned}
\end{align}
\end{lemma}
 
We say that operators that obey Eq.~\eqref{eq:aligned} with a common $P$ have aligned signs. For numbers this is the familiar fact that the absolute value of a sum equals the sum of the absolute values only if all terms have the same sign.
 
\begin{proof}
Let $\bar X$ be the integral of $X(\omega)$, let $P$ be the projector on the subspace where $\bar X$ is nonnegative, and let $Y$ be $2P-\id$. Then the trace norm of $\bar X$ equals the trace of $Y\bar X$. With the decomposition of $X(\omega)$ into positive and negative parts,
\begin{align}
\tr\left[Y X(\omega)\right] &= \norm{X(\omega)} - 2\,\tr\left[(\id-P)X^{+}(\omega)\right] \nonumber\\
&\quad - 2\,\tr\left[P X^{-}(\omega)\right].
\end{align}
We integrate over $\omega$. The left-hand side becomes the trace norm of $\bar X$, and by Eq.~\eqref{eq:equal} this equals the integral of the first term on the right. The integral of the two remaining traces therefore vanishes. Both are nonnegative, so they vanish for almost all $\omega$. A positive operator whose trace against a projector vanishes has no support in the range of that projector, which gives Eq.~\eqref{eq:aligned}.
\end{proof}
 
\begin{lemma}[Smearing]\label{lem:smear}
Let $\rho_{AB}$ be a Gaussian state that is strictly unsteerable from Alice to Bob. Then there are a positive number $\epsilon$, a probability density $q$ on $\R^{2n}$, and states $\rho_\lambda$ of $B$ such that
\begin{align}
\sigma^{\rho}_B(x) &= \int d\lambda\; q(\lambda)\, g_\epsilon(x-\lambda)\, \rho_\lambda , \label{eq:smear}
\end{align}
where $g_\epsilon$ is the centered Gaussian density with covariance matrix $\epsilon\,\id$.
\end{lemma}
 
\begin{proof}
The block $V_A$ is positive definite, so Eq.~\eqref{eq:strict} is equivalent to strict positivity of its Schur complement,
\begin{align}
V_{B|A} + \frac{i}{2}\Omega_B &> 0 .
\end{align}
For a positive number $\epsilon$ below the smallest eigenvalue of $V_A$ we define
\begin{align}
V_A' &= V_A - \epsilon\,\id, &
V'_{B|A} &= V_B - C^{T}\left(V_A'\right)^{-1} C .
\end{align}
The matrix $V_A'$ is real, symmetric, and positive definite. It is the covariance matrix of a classical Gaussian density $q$ on $\R^{2n}$. It need not obey an uncertainty relation, and the full matrix obtained from $V_{AB}$ by replacing $V_A$ with $V_A'$ is not the covariance matrix of a quantum state in general. The matrix $V'_{B|A}$ depends continuously on $\epsilon$ and tends to $V_{B|A}$. Since the inequality above is strict, we can choose $\epsilon$ so small that
\begin{align}
V'_{B|A} + \frac{i}{2}\Omega_B &\ge 0 .
\end{align}
Then $V'_{B|A}$ is the covariance matrix of a Gaussian quantum state of $B$. Let $\rho_\lambda$ be this state displaced by $C^{T}(V_A')^{-1}\lambda$. The function $q(\lambda)W_{\rho_\lambda}(x_B)$ is the Gaussian with blocks $V_A'$, $C$, and $V_B$. Convolving it in $\lambda$ with $g_\epsilon$ adds $\epsilon\,\id$ to the first block and gives $W_{AB}(x,x_B)$. Since $\sigma^{\rho}_B(x)$ is the operator with Wigner function $W_{AB}(x,\cdot)$, this is Eq.~\eqref{eq:smear}.
\end{proof}
 
\begin{lemma}[Lines and planes]\label{lem:line}
Let $\tau_{A'B'}$ be a state with $A'$ a single mode, and let $\sigma^{\tau}_{B'}(x')$ be its operator-valued Wigner function. Then the integral of $\sigma^{\tau}_{B'}(x')$ over the plane is a positive operator, and so is its integral along almost every straight line.
\end{lemma}
 
\begin{proof}
The integral over the plane is the reduced state of $B'$. For a line, let $M$ be an effect on $B'$. By Eq.~\eqref{eq:props} the function $\tr[M\sigma^{\tau}_{B'}(x')]$ is an unnormalized Wigner function of a state of $A'$. Its integral along a line is the probability density of a quadrature and is nonnegative. Since this holds for all $M$, the integral of $\sigma^{\tau}_{B'}$ along the line is a positive operator.
\end{proof}
 
\begin{proof}[Proof of Theorem~\ref{thm:nocat}]
Let $x$ in $\R^{2n}$ and $x'$ in $\R^{2}$ be the phase-space points of $A$ and $A'$ before Alice's unitary, and $y$ and $y'$ the points after it. For outcome $k$ they are related by an affine symplectic map. We write its inverse in blocks,
\begin{align}
x &= a_k y + b_k y' + e_k, &
x' &= c_k y + d_k y' + e'_k. \label{eq:blocks}
\end{align}
The unnormalized operator-valued Wigner function of the final state is, by Eq.~\eqref{eq:cov},
\begin{align}
\Omega_k(y,y') &= \mathcal{E}_k\!\left(\sigma^{\rho}_B(x)\otimes\sigma^{\tau}_{B'}(x')\right) \nonumber\\
&= \int d\lambda\; q(\lambda)\, g_\epsilon(x-\lambda)\,\Theta^{\lambda}_k(x'), \label{eq:Omega}\\
\Theta^{\lambda}_k(x') &= \mathcal{E}_k\!\left(\rho_\lambda\otimes\sigma^{\tau}_{B'}(x')\right),
\end{align}
where we used Lemma~\ref{lem:smear}. The Wigner function of Alice's modes $A$ for outcome $k$ is
\begin{align}
p_k W_k(y) &= \int dy'\; \tr\,\Omega_k(y,y'). \label{eq:Wk}
\end{align}
 
\textit{Conservation.} We apply the Wigner operator of $A'$ to both sides of the return condition, Eq.~\eqref{eq:return}. On the left the trace over $A$ is the integral over $y$, by the first relation in Eq.~\eqref{eq:cov}. This gives
\begin{align}
\sigma^{\tau}_{B'}(y') &= \sum_k \int dy\; \Omega_k(y,y') . \label{eq:returnW}
\end{align}
The map from states to operator-valued Wigner functions is linear and one to one, so Eq.~\eqref{eq:returnW} is equivalent to Eq.~\eqref{eq:return}. Inserting it into the definition of the steered negativity gives
\begin{align}
\N(\tau) &= \int dy' \norm{\sum_k \int dy\; \Omega_k(y,y')} \nonumber\\
&\le \sum_k \int dy'\, dy\, d\lambda\; q(\lambda)\, g_\epsilon(x-\lambda) \norm{\Theta^{\lambda}_k(x')} \nonumber\\
&= \sum_k \int d\lambda\; q(\lambda) \int dx' \norm{\Theta^{\lambda}_k(x')} \nonumber\\
&\le \int d\lambda\; q(\lambda)\int dx' \norm{\sigma^{\tau}_{B'}(x')} = \N(\tau). \label{eq:chain}
\end{align}
The first line is Eq.~\eqref{eq:returnW}, and all integrals are finite because $\N(\tau)$ is. In particular, $\Theta^{\lambda}_k(x')$ is integrable in $x'$ for almost all $\lambda$, and by Fubini's theorem along almost every line. The second line is the triangle inequality for the trace norm, applied for each $y'$ to the sum over $k$ and the integrals over $y$ and $\lambda$. The third line changes variables from $(y,y')$ to $(x,x')$ with unit Jacobian and integrates over $x$. The fourth line is Eq.~\eqref{eq:instr} together with the fact that $\rho_\lambda$ has unit trace norm. The chain starts and ends with the same number, so both inequalities are equalities.
 
\textit{Rigidity.} The second line of Eq.~\eqref{eq:chain} is the integral over $y'$ of an inequality that holds for each $y'$. Since the integrals are equal, the inequality is an equality for almost every $y'$. Fix such a $y'$. Lemma~\ref{lem:align}, applied to the index $\omega=(k,y,\lambda)$ and the measure with density $q(\lambda)g_\epsilon(x-\lambda)$, gives a projector $P_{y'}$ such that the operators $\Theta^{\lambda}_k(x')$ have aligned signs for almost all $(k,y,\lambda)$, with $x'$ given by Eq.~\eqref{eq:blocks}. The density $q(\lambda)g_\epsilon(x-\lambda)$ is a product of Gaussians and is strictly positive for all $y$ and $\lambda$. Sets of measure zero for this measure are therefore the sets of Lebesgue measure zero in $(y,\lambda)$. At fixed $y'$, $k$, and $\lambda$ the operator $\Theta^{\lambda}_k(x')$ depends on $y$ only through $x'$, and the map from $y$ to $x'$ in Eq.~\eqref{eq:blocks} is affine and onto a coset of the range of $c_k$. The exceptional set in $y$ is therefore the preimage of a set in that coset. A preimage under a surjective affine map has measure zero only if the set itself has measure zero in the coset. The alignment therefore holds for almost all points $x'$ of the coset and almost all $\lambda$.
 
Fix $k$. As $y$ runs over $\R^{2n}$ at fixed $y'$, the point $x'$ runs over a coset of the range of $c_k$, which is a subspace of $\R^2$. There are three cases.
 
If the range of $c_k$ is the plane, $x'$ covers the plane. The operators $\Theta^{\lambda}_k(x')$ then have aligned signs for almost all $x'$ and $\lambda$. Their integral over $x'$ is
\begin{align}
\int dx'\; \Theta^{\lambda}_k(x') &= \mathcal{E}_k\!\left(\rho_\lambda\otimes\tau_{B'}\right) \ge 0 .
\end{align}
We write $\Theta^{\lambda,\pm}_k(x')$ for the positive and negative parts and $Q$ for $\id-P_{y'}$. By Eq.~\eqref{eq:aligned} the positive parts are annihilated by $Q$ and the negative parts are left unchanged, so
\begin{align}
Q\left[\int dx'\; \Theta^{\lambda}_k(x')\right] Q &= -\int dx'\; \Theta^{\lambda,-}_k(x'). \label{eq:project}
\end{align}
The left-hand side is the compression of a positive operator and is positive. The right-hand side is minus a positive operator. Both sides therefore vanish. The integral of the positive operators $\Theta^{\lambda,-}_k(x')$ is zero, so they vanish for almost all $x'$, and $\Theta^{\lambda}_k(x')$ is positive for almost all $x'$ and $\lambda$. Only one good value of $y'$ is needed here, because the family of operators is the same for every $y'$.
 
If the range of $c_k$ is a line, $x'$ covers a straight line that depends on $y'$. The same argument applies along this line. Equation~\eqref{eq:project} holds with the integral taken along the line, and Lemma~\ref{lem:line} provides the positivity of the left-hand side. The rows of the map in Eq.~\eqref{eq:blocks} are linearly independent, so the ranges of $c_k$ and $d_k$ together span the plane, and the lines obtained for almost all $y'$ cover almost all lines of the parallel family. The exceptional lines form a set of measure zero in the plane, which does not affect the integrals in Eqs.~\eqref{eq:Omega} and \eqref{eq:Wk}. Again $\Theta^{\lambda}_k(x')$ is positive for almost all $x'$ and $\lambda$.
 
In both cases Eq.~\eqref{eq:Omega} shows that $\Omega_k(y,y')$ is positive for almost all $(y,y')$, and Eq.~\eqref{eq:Wk} shows that $W_k(y)$ is nonnegative for almost all $y$. The Wigner function of a state is continuous, so $W_k$ is nonnegative everywhere.
 
If $c_k$ vanishes, the symplectic conditions force $b_k$ to vanish as well. Alice's unitary then acts on $A$ and $A'$ separately. In Eq.~\eqref{eq:Wk} the integral over $y'$ becomes an integral over $x'$ at fixed $x$ and gives
\begin{align}
p_k W_k(y) &= \tr\,\mathcal{E}_k\!\left(\sigma^{\rho}_B(x)\otimes\tau_{B'}\right) \ge 0 .
\end{align}
The inequality holds because $\sigma^{\rho}_B(x)$ is positive. This completes the proof.
\end{proof}
 
The proof shows more than the statement. For every outcome in which Alice couples her modes to the catalyst, the whole final state of $AA'B'$ is remotely positive. A catalyst outside $\Pset$ can only be returned if such outcomes do not carry its resource.
 
The assumption of strict unsteerability enters through Lemma~\ref{lem:smear}. It guarantees that the noise that reaches the catalyst mode has full support, so that for each state $\rho_\lambda$ of Bob the point $x'$ sweeps a whole line or the whole plane. The Gaussian form of $\rho_{AB}$ is used only to obtain Eq.~\eqref{eq:smear}. This suggests that the theorem extends to remotely positive states that are not Gaussian but admit a decomposition of this form with a strictly positive kernel. We have not examined this case in detail and do not claim it.
 
On the boundary the smearing is not available, and the result cannot be obtained as a limit of strictly unsteerable states. Equation~\eqref{eq:return} is an exact condition on a given protocol and a given state, and it fails when the state is perturbed. For the state of Eq.~\eqref{eq:lossy} at transmission one half, $\sigma_B(x)$ is the density $W_A(x)$ times a coherent state whose amplitude is fixed by $x$. Bob's state and Alice's phase-space point are then locked together, the operators that are added in Eq.~\eqref{eq:chain} involve a different state of Bob for each $x'$, and the integral over $x'$ no longer factorizes as in Eq.~\eqref{eq:project}. We have neither a proof nor a counterexample for this case.
 
\textit{Preliminary measurements by Alice.} Suppose that Alice first measures some of her modes by homodyne or heterodyne detection, obtains a result $j$, and sends it to Bob, whose instrument may then depend on $j$. For each $j$ the remaining modes of Alice and the modes of Bob are in a Gaussian state. Its covariance matrix does not depend on $j$. Together with the term in Eq.~\eqref{eq:strict} it is a Schur complement of a positive definite Hermitian matrix, namely the matrix in Eq.~\eqref{eq:strict} with the covariance matrix of the measurement added on the measured modes. It is therefore positive definite, and the conditional states are strictly unsteerable. Lemma~\ref{lem:smear} applies to all of them with the same $\epsilon$. The proof above then goes through with the pair $(j,k)$ in place of $k$ and with the probability of $j$ included in the measure. Measurements by Alice at the end of the protocol need no argument, since they act on a Wigner-positive state.
 
\section{Several catalyst modes}\label{app:multi}
 
Let Alice hold $m$ catalyst modes, so that $x'$ is a point of $\R^{2m}$. The proof of Theorem~\ref{thm:nocat} uses the number of catalyst modes only in the rigidity step. There the operators $\Theta^{\lambda}_k(x')$ have aligned signs on cosets of the range of $c_k$, and we need the integral of the operator-valued Wigner function over such a coset to be positive.
 
A subspace of phase space is called coisotropic if it contains its symplectic complement. For a coisotropic subspace one can choose canonical coordinates such that the subspace consists of all coordinates of some modes and of the position coordinates of the remaining modes. The integral of the Wigner operator over a coset is then the identity on the first group of modes times a projector on a position eigenstate for each remaining mode. It is positive, and Lemma~\ref{lem:line} holds with lines replaced by cosets of the subspace. For a subspace that is not coisotropic the integral contains a Wigner operator of the untouched modes and is not positive.
 
Theorem~\ref{thm:nocat} therefore holds for $m$ catalyst modes if, for every outcome $k$, the range of the block $c_k$ in Eq.~\eqref{eq:blocks} is zero or coisotropic. This is the case when the range is the whole phase space of the catalyst modes, which is what a full-rank coupling gives when Alice has at least as many modes $A$ as catalyst modes. For a single catalyst mode every nonzero subspace is coisotropic, which is why no condition appears in the main text. The embezzling example of Remark~\ref{rem:emb} has a degenerate coupling. Only one of the $m$ catalyst modes is swapped with $A$, and the range of the block is the phase space of that one mode.
 
\end{document}